\documentclass{amsart}
\usepackage{amsmath,amscd,amsthm,amsfonts,a4wide,mathrsfs}
\usepackage[all]{xy}

\theoremstyle{plain}
\newtheorem{theorem}                 {Theorem}      [section]
\newtheorem{thm}  [theorem]  {Theorem}
\newtheorem{prop}  [theorem]  {Proposition}
\newtheorem{cor}    [theorem]  {Corollary}

\theoremstyle{definition}
\newtheorem{eg}      [theorem]  {Example}
\newtheorem{rem}       [theorem]  {Remark}
\newtheorem{defn}   [theorem]  {Definition}

\numberwithin{equation}{section}

\def \g{\mathfrak{g}}
\def \k{\mathfrak{k}}
\def \p{\mathfrak{p}}

\def \SU{{\bf SU}}
\def \S{{\bf S}}
\def \U{{\bf U}}

\def \L{\mathscr{L}} 
\def \d{\mathrm{d}}
\def \dstar{\delta}

\def \ip#1{\langle#1\rangle}

\DeclareMathOperator{\diag}{diag}

\DeclareMathOperator{\Ad}{Ad}
\DeclareMathOperator{\ad}{ad}
\DeclareMathOperator{\musicd}{\sharp}

\newcommand{\AAA}{\mathscr{A}} 
\newcommand{\BB}{\mathscr{B}} 
\newcommand{\CC}{\mathscr{C}} 
 
\newcommand{\EE}{\mathsf{E}} 
\newcommand{\rr}{\mathscr{R}} 
\newcommand{\vv}{\mathscr{V}} 
\newcommand{\hh}{\mathscr{H}} 
\newcommand{\jj}{\mathscr{J}} 
\newcommand{\vphi}{\varphi} 
\newcommand{\less}{\backslash} 
\newcommand{\R}{{\mathbb{R}}} 
\newcommand{\Z}{{\mathbb{Z}}} 
 
\newcommand{\C}{{\mathbb{C}}} 
 
\newcommand{\I}{{\mathbb{I}}} 
 
\newcommand{\CP}{{\mathbb{C}}{{P}}}

\newcommand{\beq}{\begin{equation}} 
\newcommand{\eeq}{\end{equation}} 
\newcommand{\bea}{\begin{eqnarray}} 
\newcommand{\eea}{\end{eqnarray}} 
\newcommand{\ben}{\begin{eqnarray*}} 
\newcommand{\een}{\end{eqnarray*}} 
\newcommand{\ra}{\rightarrow}

\newcommand{\cd}{\partial} 
\newcommand{\wt}{\widetilde}

\newcommand{\thet}{{\vartheta}} 
\newcommand{\eps}{{\varepsilon}}

\newcommand{\ol}{\overline}

\newcommand{\tr}{{\rm tr}} 
 
\newcommand{\id}{{\rm Id}} 
\newcommand{\Id}{{\rm Id}}

\renewcommand{\phi}{\varphi}

\allowdisplaybreaks

\begin{document}\larger[2]\setlength{\baselineskip}{1.0\baselineskip}

\title[Supercurrent coupling in the Faddeev-Skyrme model]{Supercurrent coupling 
in the Faddeev-Skyrme Model}
\author[J.M.~Speight]{J.M.~Speight}
\address{School of Mathematics, University of Leeds, Leeds, LS2 9JT} 
\email{speight@maths.leeds.ac.uk }

\begin{abstract} 
Motivated by the sigma model limit of multicomponent Ginzburg-Landau theory,
a version of the Faddeev-Skyrme model is considered in which the scalar field
is coupled dynamically to a one-form field called the supercurrent. This
coupled model is investigated in the general setting where physical space is 
an oriented Riemannian manifold and the target space is a K\"ahler manifold.
It is shown that supercurrent coupling destroys the topological stability
enjoyed by the usual Faddeev-Skyrme model, so that there can be no
globally stable knot solitons in this model. Nonetheless, local energy
minimizers may still exist. The first variation formula is
derived and used to construct three families of static solutions of the model,
all on compact domains. In particular, a coupled version of the unit-charge
hopfion on a three-sphere of arbitrary radius is found.
The second variation formula is derived, and used
to analyze the stability of some of these solutions. A family of
stable solutions is identified, though these may exist only in spaces of even 
dimension. Finally, it is shown that, in contrast to the uncoupled model, the
coupled unit hopfion on the three-sphere of radius $R$ is {\em unstable} for
all $R$. This gives an explicit, exact example of supercurrent coupling
destabilizing a stable solution of the uncoupled Faddeev-Skyrme model, and
casts doubt on the conjecture of Babaev, Faddeev and Niemi that knot solitons
should exist in the low-energy regime of 
two-component superconductors.

\end{abstract}

\maketitle

\section{Introduction}

The Faddeev-Skyrme model is a nonlinear scalar field theory which possesses
so-called knot solitons, classified topologically by their Hopf degree. 
Motivated by the sigma model limit of two-component Ginzburg-Landau
theory, Babaev, Faddeev and Niemi conjectured that such knot solitons
may exist in the low energy regime of certain exotic superconductors
\cite{babfadnie}. However, the sigma model limit contains another
dynamical field besides the usual (two-sphere valued) scalar field of
the Faddeev-Skyrme model, a one-form field which may be interpreted 
physically as the supercurrent. 
In this
paper we study this extended version of the Faddeev-Skyrme model in which 
dynamical coupling to the supercurrent is taken into account, concentrating
primarily on the (analytically more accessible) case of compact physical
domain. We will see that supercurrent coupling destroys the 
topological stability enjoyed by knot solitons,
in that configurations of arbitrarily small energy can be found in every
homotopy class. This should be contrasted with the usual Faddeev-Skyrme model
where one has the Vakulenko-Kapitanski bound on $\R^3$ \cite{kapvak},
or a {\em linear}
energy bound on compact domains \cite{spesve2}. We will develop the
variational calculus for the coupled model in a rather general geometric
setting, and use our results to show explicitly and exactly
that supercurrent coupling destabilizes the
unit  charge ``hopfion'' on the three-sphere of small radius.

Consider static two-component Ginzburg-Landau theory on physical space 
$M=\R^3$. This field theory models, among other things, certain exotic
superconducting materials, including liquid metallic hydrogen, in which there
are two different species of charge-carrying Cooper pairs
\cite{ashbabsud}. It consists of
two complex scalar fields $\psi_a:M\ra\C$, $a=1,2$, minimally coupled to
a $U(1)$ gauge connexion $A\in\Omega^1(M)$, so that the total energy 
functional is
\beq
E_{GL}=\frac12\|\d_A\psi_1\|^2+\frac12\|\d_A\psi_2\|^2+\frac12\|\d A\|^2
+\int_MV(\psi_1,\psi_2)
\eeq
where $\d_A\psi_a=\d\psi-iA\psi$, $\|\cdot\|$ denotes $L^2$ norm and
$V$ is a phenomenologically determined interaction potential whose
details depend strongly on the precise physical context. To preserve gauge
invariance, one must have 
$V(\lambda\psi_1,\lambda\psi_2)\equiv V(\psi_1,\psi_2)$ for all 
$\lambda\in \U(1)$. Babaev, Faddeev
and Niemi \cite{babfadnie} made the following interesting observation
(the essential argument appeared somewhat earlier in 
a paper of Hindmarsh \cite{hin}, and perhaps goes back further still).
Define the total condensate density $\rho=\sqrt{|\psi_1|^2+|\psi_2|^2}$ and let 
$\vphi=[\psi_1,\psi_2]:M\ra\CP^1\equiv S^2$. Note that both these fields are 
gauge invariant and that the second field
makes sense globally only if $\psi_1,\psi_2$ never vanish simultaneously, 
presumably a sensible assumption provided $V(0,0)$ is made sufficiently large.
Further, let $\rho^2C$ be the total supercurrent of the condensates, that is,
\beq
C=\frac{i}{2\rho^2}\sum_{a=1}^2(\ol{\psi}_a\d_A\psi_a-\psi_a\ol{\d_A\psi_a}),
\eeq
which is also gauge invariant.
Then the Ginzburg-Landau energy is precisely
\beq
E_{GL}=\frac18\|\rho\d\vphi\|^2+\frac12\|\d C+\frac12\vphi^*\omega\|^2+
\frac12\|\d\rho\|^2+\frac12\|\rho C\|^2+\int_M\hat{V}(\rho,\vphi)
\eeq
where $\omega$ is the K\"ahler form on $\CP^1$ (equivalently, the area form on
$S^2$). The first two terms of this energy are strongly reminiscent of the
Faddeev-Skyrme energy of a $S^2$-valued field on $M$. In fact, if
$\rho$ is constant and $C=0$, $E_{GL}$ reduces precisely to $E_{FS}$,
the Faddeev-Skyrme
energy of $\vphi$. This led Babaev, Faddeev and Niemi to suggest that the
GL model, like its truncated version, possesses knot solitons, in which the
field $\phi:\R^3\ra S^2$ has nonzero Hopf degree.

The numerical evidence concerning this claim is a little mixed, but seems, on
the whole, to be negative \cite{niepalvir,war2,jayhiesal}. (For
a comprehensive review of the current status of knot solitons in field theory,
see \cite{radvol}.)
In particular, 
a crucial role in destabilizing localized $\phi$ configurations of nonzero
Hopf degree seems to be played by the coupling of $\vphi$ and $C$. That is,
we can imagine reducing $E_{GL}$ to $E_{FS}$ by a two-step truncation. In the
first step, we impose a sigma model limit on the $\C^2$ valued field
$(\psi_1,\psi_2)$, demanding that $\rho^2=|\psi_1|^2+|\psi_2|^2=\rho_0^2=1$
everywhere,
motivated by the choice $V=\lambda(1-|\psi_1|^2-|\psi_2|^2)^2$ in the limit
of large $\lambda$, for example.
(There is no loss of generality in the choice $\rho_0=1$ since any other 
$\rho_0$ can be scaled away by a homothety of $M$.) This yields a
``halfway house'' model, which one might call the {\em supercurrent coupled
Faddeev-Skyrme model} (henceforth, the SCFS model),
\beq\label{2}
E(\vphi,C)
=\frac18\|\d\vphi\|^2+\frac12\|\d C+\frac12\vphi^*\omega\|^2+\frac12\|C\|^2
\eeq
where, for simplicity, we have assumed the potential $V(\psi_1,\psi_2)$
has $U(2)$ symmetry (so is constant on surfaces of constant $\rho$). 
Experience with similar {\em ungauged} GL models \cite{batcoosut} suggests that
this truncation is unlikely to cause trouble (that is, if $E$ has minimizers
with nonzero Hopf degree, they probably survive the thawing of 
the field $\rho$).
The second step, where we set $C$ to 
zero, is more problematic. 
This amounts to ignoring the cross term $\frac12\langle\d C,
\vphi^*\omega\rangle$ in the second term of $E(\vphi,C)$, and there does not 
seem to be a strong justification for this. 

This motivates us to study the SCFS model (\ref{2}) in detail, and compare its 
properties with those of the usual Faddeev-Skyrme model
(henceforth, FS model). Since generalization
costs no extra effort, and working in a natural geometric context often
reveals structure otherwise hidden, we will study both the SCFS and FS
 models in the
general case where $\vphi:M\ra N$, $M$ being an oriented Riemannian manifold and
$N$ a K\"ahler manifold. We will still call $C$ the supercurrent, and 
interpret $\d C+\frac12\vphi^*\omega$ as the electromagnetic field two-form
(that is, magnetic field, if $\dim M=3$). Indeed, in the case $N=\CP^{k-1}$,
with the Fubini-Study metric of unit holomorphic sectional curvature,
this model is precisely the naive
sigma model limit of $k$-component Ginzburg-Landau
theory \cite{hin}, which gives it an immediate physical interpretation.

In section \ref{bounds} we compare the best known topological
lower energy bounds for
the FS and SCFS models in the case where $M$ has dimension 3 and $N=S^2$.
Whereas the FS energy is bounded below by some power of the Hopf degree
($Q^\frac34$ for $M=\R^3$, $Q$ for $M^3$ compact), we will find that the
SCFS energy can be arbitrarily small in every (algebraically inessential)
homotopy class. This is strikingly different behaviour, and it follows that
no nontrivial global minimizer of the SCFS energy can exist in this case.
One can still hope for local minimizers, however, and to find these one must
solve the variational problem for the SCFS energy.
In section \ref{firstvariation}, we 
compute the first variation formula
for $E(\vphi,C)$, that is, the field equations that a pair $(\vphi,C)$ must
satisfy in order to be a critical point of $E$.
For a given $\vphi:M\ra N$, we will show there is at most
one $C\in\Omega^1(M)$ so that $(\phi,C)$ is a critical point of $E$.
We shall call a critical point $(\phi,C)$ {\em embedded} if $\vphi$ is
a critical point of $E_{FS}$, and refer to $(\vphi,C)$ as an {\em embedding} of $\vphi$. Clearly an embedding of $\vphi$, if it exists, is unique.
We shall construct three families of embedded
critical points, all of which are submersive (that is, 
$\d\phi_x$ is surjective at each $x\in M$). To analyze the
{\em stability} of a critical point, one must consider the second variation
formula for $E$, to which we turn in section \ref{secondvariation}.
We compute an explicit formula for the Hessian operator associated
with a critical point, and use this formula to show that the critical
point $\phi=\id:N\ra N$, $C=0$ is stable for every K\"ahler manifold $N$.

Ideally, one would like to apply the first and second variation formulae
in the case of direct physical interest, namely $M=\R^3$, $N=S^2$, but
this case of the SCFS model is (like the FS model) analytically intractable.
It would be interesting, but very challenging, and beyond the scope of the 
present work, to conduct a large-scale
numerical analysis of this problem. In lieu of hard numerics, we will
consider the case nearest to $M=\R^3$ where exact results can be obtained,
namely $M=S^3_R$, the sphere of radius $R>0$. Here, as for $M=\R^3$,
configurations
are labelled homotopically by their Hopf degree $Q$, but (unlike on
$\R^3$) an explicit 
solution of the FS model in the $Q=1$ class is known, namely the
Hopf fibration $\phi:S^3\ra S^2$\cite{war}. This solution is known to be stable
for $0<R\leq 2$ \cite{spesve1}, so it can be thought of as the
spherical analogue of the unit ``hopfion'' at least
on small spheres. In section \ref{secondvariation}, 
we will
show that this hopfion has a (necessarily unique) embedding in the SCFS
model. The associated supercurrent is homogeneous and tangent to the 
fibres of the Hopf map. We go on to construct the Hessian operator
for this embedded hopfion explicitly, and show that it has a negative eigenvalue
of total multiplicity 4, for {\em all} $R>0$. 
This, then, provides an explicit, {\em exact} example
(albeit on a compact domain) of the process whereby supercurrent coupling
destabilizes a previously stable critical point of the Faddeev-Skyrme energy.
Clearly this does not conclusively rule out the existence of stable
knot solitons in the SCFS model,
but it is one more piece of evidence in favour of scepticism.

\section{Energy bounds}\label{bounds}

We first setup some notation and conventions,
following \cite{spesve1}.
Let $(M,g)$ be a Riemannian manifold and $(N,h,J)$ be a K\"ahler manifold,
with K\"ahler form $\omega(X,Y)=h(JX,Y)$. Let
$\langle\cdot,\cdot\rangle$ denote $L^2$ inner product,
$\|\cdot\|$ $L^2$ norm, $\Omega^k(M)$ the space of smooth $k$-forms on $M$,
$\Gamma(E)$ the space of smooth sections of vector bundle $E$,
$\dstar:\Omega^k\ra\Omega^{k-1}$ the coderivative $L^2$
adjoint to $\d$, and $\flat:TM\ra T^*M$ and $\sharp=\flat^{-1}$
the musical isomorphisms induced by the metric $g$ on $M$. 
Given $\phi:M\ra N$, $\phi^{-1}TN$ will denote the
vector bundle over $M$ whose fibre above $x\in M$ is $T_{\phi(x)}N$.
The symbol $\nabla$
will denote the Levi-Civita connexion on $TM$, $TN$ or its canonical extension 
to $T^*M\otimes\vphi^{-1}TN$ (which being clear
from context). The pullback of the Levi-Civita connexion on $TN$ to
$\phi^{-1}TN$ will be denoted $\nabla^\phi$.
All maps will be assumed smooth. We shall frequently need to
refer to specfic terms in the energy functional (\ref{2}), and so define, for
any
$\vphi:M\ra N$ and $C\in\Omega(M)$,
\ben
E(\vphi,C)&=&
\frac18\|\d\vphi\|^2+\frac12\|\d C+\frac12\vphi^*\omega\|^2+\frac12\|C\|^2\\
E_1(\vphi)&=&\frac12\|\d\vphi\|^2\\
E_2(\vphi)&=&\frac12\|\vphi^*\omega\|^2\\
E_3(C)&=&\frac12\|\d C\|^2+\frac12\|C\|^2\\
E_4(\vphi,C)&=&\frac12\ip{\d C,\vphi^*\omega}\\
E_{FS}(\vphi)&=&E_1(\vphi)+E_2(\vphi),
\een
so that $E=\frac14(E_1+E_2)+E_3+E_4$.

For the rest of this section,
$N=S^2$ and $M$ has dimension 3.\, We begin by contrasting the
energy bounds for
$E$ and $E_{FS}$ in the case
$M=\R^3$. We impose the usual boundary condition on $\phi$ 
($\phi(x)\ra (0,0,1)$
as $|x|\ra \infty$) so that field configurations are homotopic
if and only if they have the same Hopf degree $Q\in\Z$. Recall
that
\beq\label{hopfdegree}
Q=\frac{1}{(4\pi)^2}\int_M A\wedge\d A
\eeq
where $A\in\Omega^1(M)$ is chosen so that $\d A=\phi^*\omega$ (note that
$A$ certainly exists since $\phi^*\omega$ is closed and $H^2(M)=0$).
Then, for $E_{FS}$ we have the well-known Vakulenko-Kapitanski bound:

\begin{thm}[Vakulenko-Kapitanski \cite{kapvak}] There exists
a constant $c>0$ such that $E_{FS}(\phi)\geq c|Q(\phi)|^{\frac34}$
for all $\phi$.
\end{thm}

The power $\frac34$ is believed to be sharp, and is certainly consistent
with numerics. The optimal constant $c$ is not known. In contrast
to this energy growth with $|Q|$, we will now see that the infimum
of $E(\phi,C)$ is zero in every homotopy class. This fact is already known
in the physics literature,
at least informally, in the multicomponent Ginzburg-Landau setting
\cite{jayhiesal,radvol}. The point is that, for a given $\phi$, we can
choose $\d C$ to exactly cancel $\frac12\phi^*\omega$ in $E$, so that
$(\phi,C)$ has no stability against Derrick scaling.

\begin{prop}\label{doh1}
 $\inf\{E(\phi,C)\: :\: Q(\phi)=n\}=0$ for all $n\in \Z$.
\end{prop}

\begin{proof}
Each degree class contains $\vphi$ with $\vphi=(0,0,1)$ outside some
closed ball $B$. For this $\vphi$,
there exists $C\in\Omega^1(\R^3)$ such that
$\d C=-\frac12\vphi^*\omega$ (since $H^2(\R^3)=0$). Since $\d\phi_x=0$
for all $x\notin{B}$, $\phi^*\omega$ vanishes on $\R^3\less{B}$. 
Without loss of generality, we may assume that $C$ itself vanishes
on $\R^3\less B$. (Assume $C$ does not vanish on $\R^3\less B$.
$C$ is closed on $\R^3\less{B}$ and hence there exists $f:\R^3\less{B}\ra
\R$ such that $C=\d f$, since $H^1(\R^3\less B)=0$. We may smoothly
extend $f$ to a function on the whole of $\R^3$. Let $C'=C-\d f$.
Then $\d C'=\d C=-\frac12\vphi^*\omega$ and $C'$ vanishes on $\R^3\less B$.)
For each $\lambda>0$ let $D_\lambda:\R^3\ra\R^3$ be the dilation map
$D_\lambda(x)=\lambda x$. Clearly $\phi\circ D_\lambda$ is homotopic to
$\phi$ and, by
 the usual Derrick argument \cite{der},
$$
E(\vphi_\lambda,C_\lambda)=\frac{1}{8\lambda}\|\d\vphi\|^2+0+
\frac{1}{2\lambda}\|C\|^2\ra 0
$$
as $\lambda\ra\infty$. 
\end{proof}

We turn now to the case where $M$ is a compact oriented 3-manifold. The 
homotopy classification of maps $M\ra S^2$ was completed by Pontrjagin.

\begin{thm}[Pontrjagin \cite{pon}] 
\label{ponthm}
Let $M$ be a compact, connected, oriented 3-manifold, and $\mu$
be a generator of $H^2(S^2;\Z)\cong\Z$ (for example,
$\mu=\omega/4\pi$ in the de Rham model). 
The homotopy classes of based maps $\phi:M\ra S^2$ fall into
disjoint families labelled by $[\phi^*\mu]\in H^2(M;\Z)$. Within the
family of classes with fixed $[\phi^*\mu]$, the classes are labelled by
elements of $H^3(M;\Z)/(2[\phi^*\mu]\smile H^1(M;\Z))$.
\end{thm}

If $H^2(M;\Z)\neq 0$,therefore,
 maps $M\ra S^2$ are not classified homotopically by a single
integer. We will be interested primarily in the case $M=S^3$, where this
complication does not arise. However, even on general $M$, there is a family
of maps which are classified homotopically by a single integer, the
{\em algebraically inessential} maps, that is, those
for which $[\phi^*\mu]=0$. These maps 
fall into homotopy classes labelled by $H^3(M;\Z)\equiv \Z$, and one
may identify this integer homotopy invariant with the Hopf degree of
$\phi$, defined as in (\ref{hopfdegree}). From the standpoint of topological
solitons, algebraically inessential maps are the most interesting, since
all other maps have the property that regular preimages $\phi^{-1}(p)\subset
M$ are nontrivial in $H_1(M)$ (one may think of such preimages as being
Poincar\'e dual to $[\phi^*\mu]\in H^2(M;\Z)$), so they are not, in a 
topological sense, spatially localized (they are ``tied'' to some nontrivial 
1-cycle in $M$). 

So, on a compact domain, provided we restrict to the case of algebraically
inessential maps (no restriction if $H^2(M;\Z)=0$), configurations
$\phi:M\ra S^2$ are still classified by their Hopf degree $Q(\phi)$.
An interesting fact, which does not seem to have been appreciated
previously, is that the Faddeev-Skyrme energy $E_{FS}$ grows at
least {\em linearly} with $|Q|$ in this case, in contrast to the $|Q|^{\frac34}$
growth found on $\R^3$. The essential proof has appeared previously for the
case $M=S^3$ \cite{spesve2}, but adapts readily to the case of general $M$.

\begin{thm}[Speight-Svensson] 
Let $M$ be a compact, connected, oriented Riemannian 3-manifold and
$\phi:M\ra S^2$ be a smooth, algebraically inessential map. Then
$$
E_{FS}(\phi)>E_2(\phi)\geq \frac{\sqrt{\lambda_1}}{32\pi^2}|Q(\phi)|
$$
where $\lambda_1>0$ is the first nonzero eigenvalue of the Laplacian
restricted to coexact one-forms on $M$.
\end{thm}

\begin{proof}
Since $\phi$ is algebraically inessential, $\phi^*\omega$ is exact, and there
exists $A\in\Omega^1(M)$ such that $\d A=\phi^*\omega$. By
Hodge decomposition, we may assume, without loss of generality
that $A$ is coexact. (Consider the
Hodge decomposition of $A$, $A=A_h+\d A_0 +\dstar A_2$,
where $A_h$ is a harmonic one-form, $A_0$ is a zero-form and 
$A_2$ a two-form. Then $\phi^*\omega=\d A =0+0+\d\dstar A_2$.) We
may also assume that $Q(\phi)>0$. (The bound is trivial if $Q=0$, and if 
$Q(\phi)<0$ then $Q(\phi\circ\rho)=-Q(\phi)$ where $\rho:M\ra M$
is an orientation reversing diffeomorphism.) Then
\ben
E_2(\phi)&=&\frac12\|\phi^*\omega\|^4=\frac12\ip{\d A,\d A}
=\frac12\ip{A,\Delta A}\geq \frac12\lambda_1\|A\|^2.
\een
But
\ben
(4\pi)^2Q(\phi)&=&\ip{A,\ast\d A}\leq \|A\|\|\d A\|=\|A\|\sqrt{2E_2(\phi)}.
\een
The result follows immediately.
\end{proof}

Note that this bound follows from a straightforward estimate of just
the Skyrme term $E_2$ in $E_{FS}$, using 
only the Cauchy-Schwartz inequality and ellipticity of $\Delta$.
By contrast, the Vakulenko-Kapitanski
bound involves a subtle trade-off between $E_1$ and $E_2$, and uses 
rather more advanced estimates from functional analysis. Note also that
under a homothety $g\mapsto R^2 g$, the spectrum of $\Delta$ scales as 
$\lambda\mapsto \lambda/R^2$, so our bound becomes trivial in the
limit where $M$ attains infinite volume (for example $M=S^3_R$, the 
three-sphere of radius $R$, in the limit $R\ra \infty$). It is
an interesting and (apparently) open question whether the power $|Q|^1$
in this bound is optimal.

Once again, the issue of primary interest in this paper is the effect
that supercurrent coupling has on this energy bound. In fact, as for the
model on $\R^3$, the bound disappears entirely in the SCFS model. 

\begin{prop}\label{doh2}
 Let $M$ be a compact, connected, oriented 3-manifold. 
Then for each $n\in\Z$,
$\inf E(\phi,C)=0$, where the infimum is taken over all smooth,
algebraically inessential maps $\phi:M\ra S^2$ with $Q(\phi)=n$, and all
$C\in\Omega^1(M)$.
\end{prop}

\begin{proof} Each algebraically inessential homotopy class contains maps
which are constant outside some (arbitrarily small) closed ball. The argument
of the proof of Proposition \ref{doh1} can be applied to such maps.
\end{proof}

So there can be no global minimizer of $E(\phi,C)$ in any nontrivial
algebraically inessential homotopy class. It does not follow that
the SCFS model can have no stable static solutions, however, since local
minima may exist. To seek them, we require the first variation formula for
$E$.

\section{The first variation}\label{firstvariation}

In this section and the next, we revert to the general setting in which $M$
is an oriented Riemannian manifold and $N$ is a K\"ahler manifold.

\begin{prop}\label{var1}
 Let $\vphi_t$ be a variation of $\vphi_0=\vphi:M\ra N$ and
$C_t$ a variation of $C_0=C\in\Omega^1(M)$. Let $X=\cd_t\vphi|_{t=0}\in
\Gamma(\vphi^{-1}TN)$ and $Y=\cd_t C_t|_{t=0}\in\Omega^1(M)$, both 
assumed to be of compact 
support (if $M$ is noncompact). Then the corresponding first variation
of $E$ is
$$
\left.\frac{dE(\vphi_t,C_t)}{dt}\right|_{t=0}=
\ip{X,-\frac14(\tau(\phi)+J\d\vphi\sharp\dstar(\vphi^*\omega+2\d C))}+
\ip{Y,\dstar(\d C+\frac12\vphi^*\omega)+C}
$$
where $\tau(\phi)$ is the
tension field of the mapping $\vphi$, that is, $\tau(\vphi)=\tr\nabla\d\vphi$.
\end{prop}

\begin{proof} Using the notation introduced in Section \ref{bounds}
above, one has
$E(\vphi_t,C_t)=\frac14(E_1(t)+E_2(t))+E_3(t)+E_4(t)$. 
It is well known \cite[p.\ 131]{ura} that
\beq
\dot{E}_1(0)=-\ip{X,\tau(\vphi)},
\eeq
and it was shown in \cite{spesve1} that
\beq
\dot{E}_2(0)=-\ip{X,J\d\vphi\sharp\dstar(\vphi^*\omega)}.
\eeq
Turning to $E_3$, one sees that
\beq
\dot{E}_3(0)=\ip{\d Y,\d C}+\ip{Y,C}=\ip{Y,\dstar\d C+C}.
\eeq
Finally, by the homotopy lemma (see, for example, \cite{spesve1}),
\beq\label{e4}
\dot{E}_4(0)=\frac12\ip{dY,\vphi^*\omega}
+\frac12\ip{\d C,\d(\vphi^*\iota_X\omega)}=
\frac12\ip{Y,\dstar(\vphi^*\omega)}+\frac12\ip{\dstar\d C,\vphi^*\iota_X\omega}.
\eeq
Now, for any $\eta\in\Omega^1(M)$,
\beq
g(\eta,\vphi^*\iota_X\omega)=(\vphi^*\iota_X\omega)(\sharp\eta)
=\omega(X,\d\vphi\sharp\eta)
=-h(X,J\d\vphi\sharp\eta).
\eeq
Hence
\beq
\ip{\eta,\vphi^*\iota_X\omega}=\int_Mg(\eta,\vphi^*\iota_X\omega)
=-\ip{X,J\d\vphi\sharp\eta}.
\eeq
Applying this in the case $\eta=\dstar\d C$, we have, from (\ref{e4})
\beq
\dot{E}_4(0)=\frac12\ip{Y,\dstar(\vphi^*\omega)}
-\frac12\ip{X,J\d\vphi\sharp\dstar\d C}.
\eeq
The result immediately follows.
\end{proof}

The first variation formula, or field equations, follow immediately from this:

\begin{cor}\label{fe}
 $(\vphi,C)$ is a critical point of $E$ if and only if
\bea
\dstar(\d C+\frac12\vphi^*\omega)+C&=&0\label{fe1}\\
\tau(\phi)-2J\d\vphi\sharp C&=&0\label{fe2}
\eea
where, once again, $\tau(\vphi)=\tr\nabla\d\vphi$.
\end{cor}

\begin{proof} For all variations $X,Y$ in Proposition \ref{var1}, we must
have $\dot{E}(0)=0$. Hence, 
by the fundamental lemma of the calculus of variations,
\bea
\dstar(\d C+\frac12\vphi^*\omega)+C&=&0,\label{eq1}\\
-\frac14(\tau(\phi)+J\d\vphi\sharp\dstar(\vphi^*\omega+2\d C))&=&0\label{eq2}.
\eea
Substituting equation (\ref{eq1}) into (\ref{eq2}) we obtain the 
pair claimed.
\end{proof}

\begin{cor}\label{emb!} 
Let $\vphi:M\ra N$. If there exists $C\in\Omega^1(M)$ such that
$(\vphi,C)$ is a critical point of $E$, then $C$ is unique.
\end{cor}

\begin{proof} Assume that $C'\in\Omega^1(M)$ also renders $(\vphi,C')$ a
critical point. Then by equation (\ref{fe1}), $C''=C-C'$ satisfies
$\dstar\d C''+C''=0$, whence $0=\ip{C'',\dstar\d C''+C''}=
\|\d C''\|^2+\|C''\|^2$. Hence $C''=0$.
\end{proof}

\begin{defn} A critical point $(\vphi,C)$ of $E$ will be called an {\em
embedding} of $\vphi$ if $\vphi$ is a critical point of $E_{FS}$. By 
Corollary \ref{emb!}, if an embedding of $\vphi$ exists, it is unique.
\end{defn}

It was shown in \cite{spesve1} that $\vphi$ is a critical point
of $E_2$ if and only if $\sharp\dstar\vphi^*\omega\in\ker\d\vphi$ everywhere.
If $\vphi$ is also harmonic, so $\tau(\vphi)=0$ 
(hence a critical point of $E_{FS}$), it is natural to seek
an embedding of $\vphi$ with $C=\mu\dstar\vphi^*\omega$, where
$\mu$ is a constant. The point is that equation (\ref{fe2}) is satisfied
automatically in this case, and we are left to check the  somewhat simpler
equation (\ref{fe1}). We will apply this idea in the next three
examples.

\begin{eg}\label{triv} 
Under what circumstances is $(\vphi,0)$ a critical point of
$E$? From Corollary \ref{fe}, 
 if and only if $\vphi$ is harmonic ($\tau(\vphi)=0$)
and coclosed ($\dstar\vphi^*\omega=0$). Such a map
is separately a critical point of
$E_1$ and $E_2$, and hence is a (very special)
critical point of $E_{FS}$.
A trivial example is the identity
map $\id:N\ra N$. This can be easily extended to projection on a
Riemannian product $\vphi=\pi:N\times K\ra N$. So $\id:S^2\ra S^2$,
$\pi:S^2\times S^1\ra S^2$ and $\pi:S^2\times\R\ra S^2$ 
embed trivially (i.e.\ with $C=0$)
in the coupled model. A less trivial family can be adapted from 
\cite{spesve2}. Let $M$ be the space of full flags in $\C^k$ and
$N$ be the Grassmannian of $l$-planes in $\C^k$. Then one has a natural
projection $\pi_l:M\ra N$ which maps each flag to its $l$-dimensional 
entry. This map is holomorphic and hence harmonic (with respect to the
usual 
$\SU(k)$ homogeneous metrics on $M$ and $N$), and is coclosed, and hence
embeds trivially in the coupled model. In particular, this 
gives a family of examples $\pi_1:M\ra\CP^{k-1}$. 
\end{eg}

\begin{eg}\label{hopfgen} 
Given the physical origin of $E(\vphi,C)$ in multicomponent
Ginzburg-Landau theory, it is natural to consider the general Hopf fibration
from $M=S^{2n+1}\subset\C^{n+1}$ to $N=\CP^n$,  given by 
$\pi:(z_1,z_2,\ldots,z_{n+1})\mapsto[z_1,z_2,\ldots,z_{n+1}]$. This map is harmonic,
and is known \cite{spesve1} to be a critical point of $E_2$,
and hence is a critical point of $E_{FS}$. We will now show that it can be
embedded in the coupled model, but in contrast to Example \ref{triv}, the
embedding has $C\neq 0$. 

Note that $\pi$ is a submersion with one dimensional fibres. At each $z
\in M
\subset\C^{n+1}$
we have an orthogonal decomposition $T_zM=\vv_z\oplus\hh_z$ where
$\vv_z=\ker\d\pi_z$ is the vertical space, tangent to the fibres, and
$\hh_z=\vv_z^\perp$ is the horizontal space, its orthogonal complement.
It is convenient to give
$N$ the Fubini-Study metric of constant holomorphic sectional curvature $4$
(in the case $n=1$, this corresponds to giving the target two-sphere radius
$\frac12$) since $\pi$ is then a Riemannian submersion, that is $\d\pi_z:\hh_z
\ra T_{\pi(z)}N$ is an isometry for all $z$. Clearly, $\vv_z$ is spanned
by the unit vector field $V(z)=iz$ (which generates the $U(1)$ action
$z\mapsto e^{it}z$ on $M\subset\C^{n+1}$). Since $\pi$ is a critical point of 
$E_2$, we know that $\d\pi\sharp\dstar(\pi^*\omega)=0$, and hence
$\dstar(\pi^*\omega)=f\flat V$ for some $f:M\ra\R$. By homogeneity of the
map, $f$ must be constant, and a short calculation based (for example)
at $z=(1,0,\ldots,0)$ shows that $f=n$. Another short calculation shows 
that $\pi^*\omega=\d\flat V$. Substituting into (\ref{fe1}) and (\ref{fe2})
one sees that $(\pi,\lambda\flat V)$ is a critical point of $E$ if (and
only if) $\lambda=\frac{n}{2n+2}$.
\end{eg}

Note that Example \ref{hopfgen} includes the standard Hopf fibration
$S^3\ra S^2$, which may be thought of as the unit charge Hopf ``soliton'' on 
$S^3$ (the special case $n=1$). We may generalize the basic Hopf fibration
in a different direction by thinking of it as the coset projection
$\SU(2)\ra \SU(2)/\S(\U(1)\times \U(1))$ onto a K\"ahler symmetric space:

\begin{eg}\label{coset}
 Let $G$ be a compact, connected, simple Lie group and $K$ be a 
compact subgroup of $G$ such that $G/K$ is an irreducible Hermitian 
symmetric space of compact type. Denote by $\g,\k$ the Lie algebras of $G,K$.
Then there is an $\Ad_K$-invariant subspace $\p$ of $\g$ such
that $[\p,\p]\subseteq\k$ and
$$
\g=\k+\p,
$$ 
this decomposition being orthogonal with respect to the Killing form.
We give $G$ the bi-invariant metric coinciding with (minus) the Killing form
at $e$, and $G/K$ the metric which makes
the homogeneous projection 
$$
\vphi:G\to G/K,\quad g\mapsto g\cdot o
$$ 
a Riemannian submersion (here $o$ denotes the identity coset in $G/K$).
Denote by $(\cdot,\cdot)$ minus the Killing form on $\g$.
The almost complex structure on $G/K$ (with respect to which the
homogeneous metric is K\"ahler) coincides with the adjoint action of an
element in the centre of $\k$. In a slight abuse of notation, we shall denote 
this element $J$. By left translation, 
we may identify vector fields on $G$ with $\g$-valued 
functions on $G$ and sections of $\vphi^{-1}T(G/K)$ with $\p$-valued
functions on $G$. The connexions on $TG$ and $\vphi^{-1}T(G/K)$ are then
\bea
\nabla_XY&=&\d Y(X)+\frac12[X,Y]\qquad(X,Y\in C^\infty(G,\g)),\label{levcit}\\
\label{pbc}\nabla_X^\vphi Y&=&\d Y(X)+[X,Y]_\p\qquad
(X\in C^\infty(G,\g),\ Y\in C^\infty(G,\p)),
\eea
respectively.

 The map $\vphi$ is well known to be harmonic, and
was shown to be a critical point of $E_2$ in \cite{spesve1}. In fact
\beq
\sharp\dstar(\vphi^*\omega)=-\frac\lambda2J
\eeq
where $\lambda$ is the eigenvalue of the Casimir operator associated with
the adjoint representation of $G$, that is,
\beq
\g\ra\g,\qquad X\mapsto-\sum_{k=1}^m[e_k,[e_k,X]]
\eeq
where $\{e_1,\ldots,e_m\}$ is an orthonormal basis for $\g$ \cite{spesve1}.
Hence, for any vector fields $X,Y$ on $G$,
\ben
\d(\dstar\vphi^*\omega)(X,Y)&=&-\frac\lambda2\{
X(J,Y)-Y(J,X)-(J,[X,Y])\}\\
&=&-\frac\lambda2\{(\nabla_XJ,Y)-(\nabla_YJ,X)\}\\
&=&-\frac\lambda2\{(\d J(X),Y)-(\d J(Y),X)+\frac12([X,J],Y)-\frac12([Y,J],X)\}\\
&=&-\frac\lambda2\{0-0-\frac12(\ad_JX,Y)+\frac12(\ad_JY,X)\}\\
&=&\frac\lambda2\vphi^*\omega(X,Y),
\een
where we have used the fact that $J$ is a {\em constant} mapping $M\ra\g$
(so $\d J=0$). Hence
\beq\label{ddom}
\d\dstar\vphi^*\omega=\frac\lambda2\vphi^*\omega.
\eeq

We now seek an embedding of $\vphi$ with $C=\mu\dstar\vphi^*\omega$.
As remarked previously, (\ref{fe2}) is satisfied automatically, and, in light
of (\ref{ddom}), we see that (\ref{fe1}) is satisfied if and only if
$\mu=-(\lambda+2)^{-1}$ (which always exists since the Casimir is a positive
operator). In particular, let $G=\SU(n+1)$ and $K=\S(\U(1)\times \U(n))$. Then
this gives an embedding of the projection $\vphi:\SU(n+1)\ra G/K\equiv
\CP^n$. Note that the vertical space of the fibration $\vphi$
has dimension $n^2$ in this case, rather than $1$, as in Example
\ref{hopfgen}. From the point of view of satisfying
equation (\ref{fe2}), it looks like choosing $C$ parallel to 
$\dstar\vphi^*\omega$ is needlessly restrictive ($\sharp C$ need only be 
restricted to the $n^2$-dimensional vertical space). However, by
Corollary \ref{emb!}, we know that no alternative embedding of $\vphi$
exists, despite the apparent extra freedom.
\end{eg}

\section{Second variation: the Hessian}\label{secondvariation}

Let $(\vphi,C)$ be a critical point of $E$ and $(\vphi_{s,t},C_{s,t})$
be a two-parameter variation of $(\vphi,C)=(\vphi_{0,0},C_{0,0})$. The
{\em Hessian} of $E$ at $(\vphi,C)$ is
$$
H_{(\vphi,C)}((X,Y),(\hat{X},\hat{Y}))=
\left.\frac{\cd^2E(\vphi_{s,t},C_{s,t})}{\cd s\cd t}\right|_{s=t=0}
$$
where
\ben
X=\cd_s\vphi_{s,t}|_{s=t=0},\quad
\hat{X}=\cd_t\vphi_{s,t}|_{s=t=0}&\in&\Gamma(\vphi^{-1}TN),\\
Y=\cd_sC_{s,t}|_{s=t=0},\quad
\hat{Y}=\cd_tC_{s,t}|_{s=t=0}&\in&\Omega^1(M).
\een
Let $\EE$ denote the vector bundle $\vphi^{-1}TN\oplus T^*M$ over $M$. Then
$H_{(\vphi,C)}$ is a symmetric bilinear form on $\Gamma(\EE)$. The critical 
point $(\vphi,C)$ is {\em stable} if the associated
quadratic form is non-negative, that is,
$$
H_{(\vphi,C)}((X,Y),(X,Y))\geq 0\qquad\mbox{for all $(X,Y)
\in\Gamma(\EE)$.}
$$
Otherwise, $(\vphi,C)$ is unstable. The {\em index} of an unstable critical
point is the dimension of the largest subspace of $\Gamma(\EE)$
on which the quadratic form
is negative. We will see that there is a self-adjoint linear differential
operator $\hh_{(\vphi,C)}:\Gamma(\EE)\ra\Gamma(\EE)$, which we will call
the {\em Hessian operator}, such that
$$
H_{(\vphi,C)}((X,Y),(\hat{X},\hat{Y}))=\ip{(X,Y),\hh_{(\vphi,C)}(\hat{X},\hat{Y})}.
$$
To show that a critical point in unstable, it suffices to exhibit a
negative eigenvalue of its Hessian operator.

The aim of this section is to compute $\hh_{(\vphi,C)}$.
We start by recalling the expressions for the analogous operators for
critical points of the functionals $E_1$ and $E_2$, denoted
$\jj$ and $\L$ respectively. The first of these is conventionally
called the {\em Jacobi operator}. Since $E_2$ can be naturally viewed
as a symplectic analogue of $E_1$, we shall call $\L$ the
{\em symplectic Jacobi operator}.

\begin{prop}[Smith, \cite{smi}] \label{jjprop}
Let $\vphi:M\ra N$ be a critical point of 
$E_1=\frac12\|\d\vphi\|^2$. Then the Hessian of $E_1$ at $\vphi$ is
$\ip{X,\jj_\vphi Y}$ where the Jacobi operator is
$$
\jj_\vphi Y=\ol{\Delta}_\vphi Y-\rr_\vphi Y.
$$
The terms of 
$\jj_\vphi$ are defined as follows. Let $\{e_1,\ldots, e_m\}$ be a local
orthonormal frame on $M$ and $R$ be the curvature of $\nabla$ on $N$. 
Then the rough Laplacian is
$$
\ol{\Delta}_\vphi Y=\sum_{k=1}^m(\nabla_{e_k}\nabla_{e_k}Y-\nabla_{\nabla_{e_k}e_k}Y)
$$
and 
$$
\rr_\vphi Y=\sum_{k=1}^mR(Y,\d\vphi E_k)\d\vphi E_k.
$$
\end{prop}

\begin{prop}[Speight-Svensson, \cite{spesve1}]\label{lprop}
Let $\vphi:M\ra N$ be a critical point of 
$E_2=\frac12\|\vphi^*\omega\|^2$. Then the Hessian of $E_2$ at $\vphi$ is
$\ip{X,\L_\vphi Y}$ where the symplectic Jacobi operator is
$$
\L_\vphi Y=-J\bigg(\nabla^\vphi_{Z_\vphi}Y+
\d\vphi(\musicd\dstar\d\vphi^*\iota_Y\omega)\bigg)\ \text{ and }\  
Z_\vphi=\musicd\dstar\vphi^*\omega.
$$
\end{prop}

\begin{prop}\label{hessprop}
 Let $(\vphi,C)$ be a critical point of $E$. Then the Hessian 
operator for  
$E$ at $(\vphi,C)$ is 
$$
\hh_{(\vphi,C)}\left[\begin{array}{c}X\\Y\end{array}\right]
=\left[\begin{array}{c}
\frac14(\jj_\vphi X+\L_\vphi X)-\frac12J\nabla^\vphi_{\sharp\dstar\d C}X
-\frac12J\d\vphi\sharp\dstar\d Y\\
\dstar\d(Y+\frac12\vphi^*\iota_X\omega)+Y\end{array}\right].
$$
Recall that $X$ is a section of $\vphi^{-1}TN$ and $Y$ is a one-form on $M$.
\end{prop}

\begin{proof} Let $(\vphi_{s,t},C_{s,t})$ be a smooth variation of critical point
$(\vphi,C)=(\vphi_{0,0},C_{0,0})$, and $(X,Y),(\hat{X},\hat{Y})\in\Gamma(\EE)$
be its infinitesimal generators. Then, since $(\vphi,C)$ is critical,
\ben
\left.\frac{\cd^2 E(\vphi_{s,t},C_{s,t})}{\cd s\, \cd t}\right|_{s=t=0}&=&
-\frac14\ip{X,\cd_t\tau(\vphi_{0,t})|_{t=0}}
-\frac14\ip{X,\cd_t(J\d\vphi_{0,t}\sharp\dstar(\vphi^*\omega))|_{t=0}}\\
&&
+\ip{Y,\cd_t(\dstar\d C_{0,t}+C_{0,t})|_{t=0})}
+\frac12\ip{Y,\cd_t\dstar(\vphi_{0,t}^*\omega)|_{t=0}}\\
&&
-\frac12\ip{X,\cd_t(J\d\vphi_{0,t}\sharp\dstar\d C_{0,t})|_{t=0}}\\
&=&\frac14\ip{X,(\jj_\vphi+\L_\vphi)\hat{X}}
+\ip{Y,\dstar\d\hat{Y}+\hat{Y}}
+\frac12\ip{X,\dstar\d(\vphi^*\iota_{\hat{X}}\omega)}\\
&&
-\frac12\ip{X,J\cd_t(\d\vphi_{0,t}\sharp\dstar\d C_{0,t})|_{t=0}}
\een
by Propositions \ref{jjprop} and \ref{lprop}, and the Homotopy Lemma.
Let $(\vphi_t,C_t)=(\vphi_{0,t},C_{0,t})$, and $F:(-\eps,\eps)\times M\ra N$
be the total map $F(t,x)=\vphi_t(x)$. Then
\ben
\cd_t(\d\vphi_{t}\sharp\dstar\d C_{t})|_{t=0}&=&
\left(\nabla^F_{\cd/\cd t}\d F\right)_{t=0}\sharp\dstar\d C+
\d\vphi\sharp\dstar\d(\cd_tC_t)|_{t=0}\\
&=&\left(\nabla^F_{\sharp\dstar\d C}\d F\right)_{t=0}\cd/\cd t+
\d\vphi\sharp\dstar\d\hat{Y}\\
&=&\nabla^\vphi_{\sharp\dstar\d C}\hat{X}
-\d\vphi\left(\nabla^{(-\eps,\eps)\times M}_{\sharp\dstar\d C_t}\frac{\cd\: }{\cd t}
\right)_{t=0}+\d\vphi\sharp\dstar\d\hat{Y}\\
&=&\nabla^\vphi_{\sharp\dstar\d C}\hat{X}+\d\vphi\sharp\dstar\d\hat{Y}
\een
since the Levi-Civita connexion on $N$ is torsionless, and $(-\eps,\eps)\times
M$ has the product metric. The result immediately follows.
\end{proof}

\begin{rem}\label{selfadjoint}
 A nontrivial check on this formula is that $\hh$ should
be self-adjoint with respect to the $L^2$ inner product on $\Gamma(\EE)$.
This is clear, provided that the operators
\ben
&&\AAA:\Omega^1(M)\ra \Gamma(\vphi^{-1}TN)\qquad 
\AAA:Y\mapsto -J\d\vphi\sharp\dstar\d Y\\
&&\BB:\Gamma(\vphi^{-1}TN)\ra\Omega^1(M)\qquad
\BB:X\mapsto \dstar\d(\vphi^*\iota_X\omega)
\een
are an adjoint pair, that is, $\BB^\dagger=\AAA$. Let us check this:
\ben
\ip{Y,\BB X}&=&\ip{\dstar\d Y,\vphi^*\iota_X\omega}
=\int_M\vphi^*\iota_X\omega(\sharp\dstar\d Y)
=\int_M\omega(X,\d\vphi\sharp\dstar\d Y)\\
&=&\ip{JX,\d\vphi\sharp\dstar\d Y}=\ip{X,\AAA Y}
\een
as required.
\end{rem}

The formula for the total Hessian operator is, admittedly, rather complicated.
However, it is not so complicated as to be unusable, as we shall now see.

\begin{cor} $(\Id,0)$ is a stable critical point of $E$, where $\Id:N\ra N$ 
is the identity map on any K\"ahler manifold.
\end{cor}

\begin{proof} Certainly $(\Id,0)$ is a critical point
(Example \ref{triv}). 
It is known \cite[p.\ 172]{ura}
that 
$$
\ip{X,\jj_{\vphi}X}\geq 0
$$ 
for any holomorphic 
map $\vphi$ between K\"ahler manifolds, and hence, in particular, for
$\vphi=\Id$. In the case of $\vphi=\Id$, $\vphi^*\omega=\omega$, which is 
coclosed, and one has a canonical identification
of $\vphi^{-1}TN$ with $TM=TN$, so that $\vphi^*(\iota_X\omega)=\flat JX$,
and one finds that $\ip{X,\L_{\Id}X}=\|\d\flat JX\|^2$
\cite{spesve1}. Hence
\ben
\ip{(X,Y),\hh(X,Y)}&=&\frac14\ip{X,\jj_{\Id}X}+\frac14\|\d\flat JX\|^2
+\|\d Y\|^2+\|Y\|^2+\ip{\d Y,\d\flat JX}\\
&\geq&\|\d\flat J\frac{X}{2}\|^2+2\ip{\d Y,\d\flat J\frac{X}{2}}+\|\d Y\|^2
\geq 0
\een
\end{proof}

As a less trivial application of the second variation formula, we can
use it to show that the embedded Hopf fibration $S^3\ra S^2$ is unstable.
This is the subject of the next section.

\section{Stability of the embedded Hopf map $S^3\ra S^2$}\label{hopf}

In this section we will analyze in detail the Hessian operator for
the embedded Hopf fibration $S^3_R\ra S^2$, where $S^3_R$ is the 3-sphere of
radius $R$ and $S^2$ is the unit 2-sphere. It is convenient to identify the
domain with $G=\SU(2)$, and the codomain with $G/K$, where
$K=\{\diag(\lambda,\ol{\lambda})\: :\:\lambda\in \U(1)\}$, so that 
$\vphi:G\ra G/K$ is the coset projection, and use the machinery outlined in
Example \ref{coset}.
Let $\thet_1,\thet_2,\thet_3$ be the usual basis of left-invariant vector
fields on $G$ (coinciding at $e$ with $\frac{i}{2}\tau_a$, $a=1,2,3$, where
$\tau_a$ are the Pauli spin matrices). Note that $\k=\{\thet_3\}_\R$ and
$\p=\{\thet_1,\thet_2\}_\R$, where $\{\cdot\}_\R$ denotes linear span.
The almost complex structure on $G/K$ coincides with the adjoint action
of $J=\thet_3$ on $\p$, and the pullback of the K\"ahler form $\omega$
to $G$ is $\vphi^*\omega=-\sigma_1\wedge\sigma_2$, where
$\{\sigma_a\}$ is the coframe dual to $\{\thet_a\}$. 

Unlike Example \ref{coset}, we will not give $G$ the metric coinciding at
$e$ with minus the Killing form, nor $G/K$ the metric which renders
$\vphi$ a Riemannian submersion. Rather, $G$ and $G/K$ are given the round
metrics of radius $R$ and $1$, respectively. Making use of the canonical
identifications
$TG\equiv G\times\g$ and $\vphi^{-1}(TG/K)\equiv G\times\p$, 
this amounts to declaring $\thet_a\in TG$, $a=1,2,3$,
 to be orthogonal vector fields
of length $\frac{R}{2}$, and $\thet_1,\thet_2$ to be orthonormal
sections of $T(G/K)$. It follows that $\dstar\vphi^*\omega=-\frac{4}{R^2}
\sigma_3$, and that the pair $(\vphi,C)$ satisfies (\ref{fe1}),(\ref{fe2})
if and only if
\beq
C=\frac{2}{4+R^2}\sigma_3.
\eeq
The pair $(\vphi,C)$ is the analogue for the supercurrent coupled
Faddeev-Skyrme model on $S^3_R$ of the unit charge Hopf soliton on
$\R^3$. We note in passing that the magnetic field of this 
coupled ``hopfion'' is tangent to the fibres of $\vphi$,
\beq
B=\sharp * (\d C +\frac12\vphi^*\omega)=-\frac{4}{R(4+R^2)}\thet_3.
\eeq

 Making use of the
identifcations $\vphi^{-1}T(G/K)\equiv G\times\p$ and
$T^*G\equiv G\times\g^*$, we can identify any smooth section of $\EE$
with a quintuple of smooth functions $f_i:G\ra\R$, $i=1,2,\ldots, 5$, by
\beq
(X,Y)=(f_1\thet_1+f_2\thet_2,f_3\sigma_1+f_4\sigma_2+f_5\sigma_3).
\eeq
If we similarly identify the section $\hh(X,Y)$ with a mapping
$G\ra\R^5$, 
 the Hessian operator for $E$ at $(\vphi,C)$
is represented by a $5\times 5$ matrix of differential operators
acting on real-valued functions on $G$. Similarly, the Jacobi
and symplectic Jacobi operators 
$\jj,\L:\Gamma(\vphi^{-1}TN)\ra\Gamma(\vphi^{-1}TN)$ are represented by
$2\times 2$ matrices of differential operators. We shall denote the matrix 
representing $\hh$  by
$(\hh)$. Similarly, the $2\times 2$ matrix of differential operators
representing $\jj$ will be denoted $(\jj)$, and so on.
From Proposition \ref{hessprop}, we see that $(\hh)$ has the block 
structure
\beq\label{block}
(\hh)=\begin{pmatrix}\frac14(\jj)+\frac14(\L)+(\CC)&\frac12(\AAA)\\
\frac12(\BB)&(\dstar\d)+\I_3\end{pmatrix},
\eeq
where $\AAA,\BB$ were defined in Remark \ref{selfadjoint}, and 
\beq
\CC:\Gamma(\vphi^{-1}TN)\ra\Gamma(\vphi^{-1}TN)\qquad
\CC:X\mapsto-\frac12J\nabla^\vphi_{\sharp\dstar\d C}X.
\eeq
 
We seek an explicit formula for $(\hh)$.
 The $(\jj)$ and $(\L)$ parts of this formula are already known
\cite{ura2,spesve1}, in principle. These papers each equip $S^3$ with a fixed
radius ($R=2\sqrt{2}$ for $\jj$ in \cite{ura2} and $R=2$ for $\L$ in 
\cite{spesve1}),
however, so to make use of their results we must determine how $\jj$ and
$\L$ scale under homotheties of $M$. For completeness, we shall
simultaneously consider their scaling under homotheties of $N$ also.

\begin{prop}\label{homothety}
 Let $\vphi:(M^m,g)\ra (N^n,h)$ be harmonic with Jacobi
operator $\jj$. If $\wt{g}=R_1^2\, g$, $\wt{h}=R_2^2\, h$, where 
$R_1,R_2>0$ are constants, then the Jacobi operator of $\vphi$ as a
harmonic map $(M,\wt{g})\ra (N,\wt{h})$ is
$$
\wt{\jj}=R_1^{-2}\jj.
$$
Similarly, let $\vphi:(M^m,g)\ra (N^n,h)$ be a critical point of $E_2$
with symplectic Jacobi operator $\L$. Then the symplectic Jacobi operator
of $\vphi$ as a critical point $(M,\wt{g})\ra (N,\wt{h})$ is
$$
\wt{\L}=R_1^{-4}R_2^2\L.
$$

\end{prop}

\begin{proof}
Let $\vphi_{s,t}:M\ra N$ be
a smooth two-parameter variation of $\vphi=\vphi_{0,0}$, with
$\cd_s\vphi_{s,t}|_{s=t=0}=X$, 
$\cd_t\vphi_{s,t}|_{s=t=0}=Y\in\Gamma(\vphi^{-1}TN)$, and note that
the Hodge isomorphism $\Omega^p(M)\ra\Omega^{m-p}(M)$
scales as 
$\wt{*}=R_1^{m-2p}*$
under the homothety $g\ra\wt{g}$. We have, in obvious
notation,
\ben
\wt{E}_1(\vphi_{s,t})&=&R_1^{m-2}R_2^2E_1(\vphi_{s,t})\\
\Rightarrow\quad
\left.\frac{\cd^2\wt{E}_1(\vphi_{s,t})}{\cd s\, \cd t}\right|_{s=t=0}
&=&R_1^{m-2}R_2^2\int_Mh(X,\jj Y)*1
=\int_M\wt{h}(X,R^{-2}_1\jj Y)\wt{*}1. 
\een
Similarly, since $\wt{\omega}=R_2^2\omega$,
\ben
\wt{E}_1(\vphi_{s,t})&=&\int_M\vphi_{s,t}^*\wt{\omega}\wedge
\wt{*}\vphi_{s,t}^*\wt{\omega}
=R_1^{m-4}R_2^4E_2(\vphi_{s,t})\\
\Rightarrow\quad
\left.\frac{\cd^2\wt{E}_2(\vphi_{s,t})}{\cd s\, \cd t}\right|_{s=t=0}
&=&R_1^{m-4}R_2^4\int_Mh(X,\L Y)*1
=\int_M\wt{h}(X,R^{-4}_1R_2^2\L Y)\wt{*}1. 
\een
\end{proof}

It follows from the results of \cite{ura2} and Proposition \ref{homothety}
that, for the Hopf fibration $S^3_R\ra S^2$,
\beq\label{Jmat}
(\jj)=\frac{4}{R^2}\left(\begin{array}{cc}
-(\thet_1^2+\thet_2^2+\thet_3^2)&-2\thet_3\\
2\thet_3&-(\thet_1^2+\thet_2^2+\thet_3^2)
\end{array}\right).
\eeq
 Similarly, from
the calculation
in \cite{spesve1} and Proposition \ref{homothety}, one sees that
\beq\label{Lmat}
(\L)=\frac{16}{R^4}\begin{pmatrix}
-\thet_1^2-\thet_3^2 & -\thet_3-\thet_1\thet_2 \\ 
\thet_3-\thet_2\thet_1 & -\thet_2^2-\thet_3^2\end{pmatrix}.
\eeq
To complete the computation of the top-left $2\times 2$ block of $(\hh)$,
we need the matrix representing $\CC$. Now
\beq
\sharp\dstar\d C=\frac{32}{R^4(4+R^2)}\thet_3
\eeq
and, from equation (\ref{pbc}),
\beq
\nabla^\vphi_{\thet_3}(f_1\thet_1+f_2\thet_2)=\thet_3(f_1)\thet_1
+\thet_3(f_2)\thet_2-f_1\thet_2+f_2\thet_1.
\eeq
Hence, since $J\thet_1=[\thet_3,\thet_1]=-\thet_2$,
\beq\label{Cmat}
(\CC)=\frac{16}{R^4(4+R^2)}\begin{pmatrix}
1&-\thet_3\\ \thet_3& 1
\end{pmatrix}.
\eeq

A short calculation of $\delta\d (f\sigma_a)$ using
$\d\sigma_1=\sigma_2\wedge\sigma_3$,
$*\sigma_1\wedge\sigma_2=\frac{2}{R}\sigma_3$, and cyclic
permutations, shows that 
\beq\label{ddmat}
(\dstar\d)=\frac{4}{R^2}\begin{pmatrix}
1-\thet_2^2-\thet_3^2&\thet_2\thet_1-2\thet_3&\thet_3\thet_1+2\thet_2\\
\thet_1\thet_2+2\thet_3&1-\thet_1^2-\thet_3^2&\thet_3\thet_2-2\thet_1\\
\thet_1\thet_3-2\thet_2&\thet_2\thet_3+2\thet_1&1-\thet_1^2-\thet_2^2
\end{pmatrix}.
\eeq

It remains to compute the top-right block $(\AAA)$ and the bottom-left
block $(\BB)$ of $(\hh)$. These form an $L^2$ adjoint pair, but care
must be taken when using this fact, since our basis for $\EE$ is
{\em not} orthonormal. Recall that $\d\vphi$ coincides with
orthogonal projection $\g\ra\p$, so, making use of  the formula 
(\ref{ddmat}) for 
$(\dstar\d)$ above, one finds,
\bea\label{Amat}
(\AAA)&=&-\frac12\begin{pmatrix}0&1\\-1&0\end{pmatrix}
\frac{4}{R^2}\begin{pmatrix}1&0&0\\0&1&0\end{pmatrix}(\dstar\d)\nonumber\\
&=&\frac{8}{R^4}\begin{pmatrix}
-\thet_1\thet_2-2\thet_3&\thet_1^2+\thet_3^2&-\thet_3\thet_2+2\thet_1\\
-\thet_2^2-\thet_3^2&\thet_2\thet_1-2\thet_3&\thet_3\thet_1+2\thet_2
\end{pmatrix}.
\eea
Finally, the map $\Lambda:\Gamma(\vphi^{-1}TN)\ra\Omega^1(M)$,
$\Lambda:X\mapsto\vphi^*\iota_X\omega$ has matrix representative
\beq
(\Lambda)=\begin{pmatrix}0&1\\-1&0\\0&0\end{pmatrix},
\eeq
which, together with (\ref{ddmat}), yields
\beq\label{Bmat}
(\BB)=\frac12(\dstar\d)(\Lambda)=\frac{2}{R^2}
\begin{pmatrix}
-\thet_2\thet_1+2\thet_3&-\thet_2^2-\thet_3^2\\
\thet_1^2+\thet_3^2&\thet_1\thet_2+2\thet_3\\
-\thet_2\thet_3+2\thet_1&\thet_1\thet_3-2\thet_2
\end{pmatrix}
\eeq
Assembling (\ref{Jmat}), (\ref{Lmat}), (\ref{Cmat}), (\ref{ddmat}), 
(\ref{Amat}),
(\ref{Bmat}) with (\ref{block}) we obtain the explicit formula for
$(\hh)$, which we will not reproduce here.

We now have an expression for $\hh$ as a matrix of differential
operators acting on functions on $G$.
To proceed further, we must choose a basis for $L^2(G,\C)$, and for this 
purpose we appeal to the Peter-Weyl theorem \cite[p.\ 17]{petwey}.
According to this, the matrix elements of the finite-dimensional
irreducible unitary representations of $G$ form an orthonormal basis for
$L^2(G,\C)$. Now $\hh$ commutes with the obvious action of $G$ on
$\Gamma(\EE)$. Hence, for a given fixed
irreducible representation of $G$,
 the finite dimensional subspace of $\Gamma(\EE)$
in which each of $f_1,f_2,\ldots, f_5:G\ra\C$ lies in the span of
the matrix elements of that representation is invariant under $\hh$.
So we can break down the (infinite dimensional) eigenvalue problem for
$\hh$ into an infinite sequence of {\em finite} dimensional eigenvalue
problems, indexed by the irreducible unitary representations of $G$.
In this case, $G=\SU(2)$, whose irreducible unitary representations
are indexed by a non-negative integer $n$, called the
weight
 (loosely speaking, twice
the ``spin'' associated with the representation). For fixed $n$, we
can construct $(n+1)\times(n+1)$ antihermitian matrices representing
$\thet_1,\thet_2,\thet_3$, and hence construct a $(5n+5)\times (5n+5)$
matrix $\hh^{(n)}$ representing $\hh$ acting on the weight $n$
invariant subspace of $\Gamma(\EE)$. The machinery for
dealing with the general $n$ case is developed in \cite{spesve1}.
However, for our purposes, we will need only the fundamental representation
$n=1$. In this case, we simply replace $\thet_a$ by $\frac{i}{2}\tau_a$
in $(\hh)$ to obtain a $10\times 10$ complex matrix $\hh^{(1)}$, each 
of whose entries is a rational function of $R$, the radius of $M$.
The eigenvalue problem for $\hh^{(1)}$ can be solved exactly, by Maple, for 
example. One finds that $\hh^{(1)}$ has an eigenvalue of multiplicity two,
namely
\ben
\lambda(R)&=&\frac{p(R)-\sqrt{p(R)^2+3840R^4+1536R^6+208R^8+16R^{10}}}{8R^4(4+R^2)},\\
\mbox{where}\quad p(R)&=&48+144{R}^{2}+51{R}^{4}+4{R}^{6}.
\een
 Clearly $\lambda(R)<0$ for all $R>0$. Since
our basis vectors for $L^2(G,\C)$ are {\em matrix} elements, they are
indexed by an ordered {\em pair} of indices, taking values in
$\{1,\ldots, n+1\}$. It follows that an eigenvalue of $\hh^{(n)}$ of
multiplicity $m$
is an eigenvalue of $\hh$ (restricted to the weight $n$ invariant
subspace of $\Gamma(\EE)$) of multiplicity $(n+1)m$. Hence,
the {\em index} (dimension of the sum of the negative eigenspaces)
of $\hh$ is at least $2\times(1+1)=4$.
The eigenvalue problem for $\hh^{(n)}$ for
$1\leq n\leq 10$ has been solved numerically for various $R$, no further
negative eigenvalues being found. It seems likely, therefore, that
the index of $\hh$ is exactly $4$. In any case, the $n=1$ calculation
outlined above establishes rigorously that the embedded Hopf fibration
is an {\em unstable} critical point of $E$, for all choices of the
radius $R$ of $S^3_R$.

\section*{Acknowledgements}

The author wishes to thank Martin Svensson for valuable conversations,
and acknowledges financial support from the UK Engineering and Physical
Sciences Research Council.

\end{document}